\documentclass[11pt]{article}
\usepackage[a4paper]{geometry}
\usepackage{amsfonts, amsmath, amssymb, amsthm, graphicx, caption, authblk, multirow, makecell, framed, float, tikz, xcolor, enumitem}
\usetikzlibrary{shapes,arrows,positioning,decorations.pathreplacing}
\theoremstyle{definition}
\newtheorem{theorem}{Theorem}
\newtheorem{definition}{Definition}
\newtheorem{example}{Example}

\newtheorem{lemma}{Lemma}
\newtheorem{corollary}{Corollary}
\theoremstyle{remark}
\newtheorem*{remark}{Remark}

\begin{document}
% Define block styles
\tikzstyle{line} = [draw]
\tikzstyle{arrow} = [draw, ->]
\tikzstyle{vertex} = [draw, circle, minimum size=1cm]
\title{Stable Noncrossing Matchings}
\author[1]{Suthee Ruangwises\thanks{\texttt{ruangwises.s.aa@m.titech.ac.jp}}}
\author[1]{Toshiya Itoh\thanks{\texttt{titoh@c.titech.ac.jp}}}
\affil[1]{Department of Mathematical and Computing Science, Tokyo Institute of Technology, Tokyo, Japan}
\date{}
\maketitle

\begin{abstract}
Given a set of $n$ men represented by $n$ points lying on a line, and $n$ women represented by $n$ points lying on another parallel line, with each person having a list that ranks some people of opposite gender as his/her acceptable partners in strict order of preference. In this problem, we want to match people of opposite genders to satisfy people's preferences as well as making the edges not crossing one another geometrically. A \textit{noncrossing blocking pair} w.r.t. a matching $M$ is a pair $(m,w)$ of a man and a woman such that they are not matched with each other but prefer each other to their own partners in $M$, and the segment $(m,w)$ does not cross any edge in $M$. A \textit{weakly stable noncrossing matching} (WSNM) is a noncrossing matching that does not admit any noncrossing blocking pair. In this paper, we prove the existence of a WSNM in any instance by developing an $O(n^2)$ algorithm to find one in a given instance.

\end{abstract}

\section{Introduction}
The stable marriage problem is one of the most classic and well-studied problems in the area of matching under preferences, with many applications in other fields including economics \cite{gusfield,roth}. We have a set of $n$ men and a set of $n$ women, with each person having a list that ranks some people of opposite gender as his/her acceptable partners in order of preference. A \textit{matching} is a set of disjoint man-woman pairs. A \textit{blocking pair} w.r.t. a matching $M$ is a pair of a man and a women that are not matched with each other in $M$ but prefer each other to their own partners. The goal is to find a \textit{stable matching}, a matching that does not admit any blocking pair.

On the other hand, the noncrossing matching problem is a problem in the area of geometric matching. We have a set of $2n$ points lying on two parallel lines, each containing $n$ points. We also have some edges joining vertices on the opposite lines. The goal is to select a set of edges that do not cross one another subject to different objectives, e.g. maximum size, maximum weight, etc.

In this paper, we study a problem in geometric matching under preferences. In particular, we investigate the problem of having $n$ men and $n$ women represented by points lying on two parallel lines, with each line containing $n$ people of one gender. Each person has a list that ranks some people of opposite gender in strict order of preference. A \textit{noncrossing blocking pair} w.r.t. a matching $M$ is a blocking pair w.r.t. $M$ that does not cross any edge in $M$. Our goal is to find a noncrossing matching that does not admit any noncrossing blocking pair, called a \textit{weakly stable noncrossing matching} (WSNM).

Note that the real-world applications of this geometric problem are more likely to involve immovable objects, e.g. construction of noncrossing bridges between cities on the two sides of a river, with each city having different preferences. In this paper, however, we keep the terminologies of men and women used in the original stable marriage problem in order to understand and relate to the original problem more easily.

\subsection{Related Work}
The stable marriage problem was first introduced by Gale and Shapley \cite{gale}. They proved that a stable matching always exists in an instance with $n$ men and $n$ women, with each person's preference list containing all $n$ people of opposite gender and not containing ties. They also developed an $O(n^2)$ algorithm to find a stable matching in a given instance. Gusfield and Irving \cite{gusfield} later showed that the algorithm can be adapted to the setting where each person's preference list may not contain all people of opposite gender. The algorithm runs in $O(m)$ time in this setting, where $m$ is the total length of people's preference lists. Gale and Sotomayor \cite{gale2} proved that in this modified setting, a stable matching may have size less than $n$, but every stable matching must have equal size. Irving \cite{irving2} then generalized the notion of a stable matching to the case where ties are allowed in people's preference lists. He introduced three types of stable matchings in this setting: \textit{weakly stable}, \textit{super-stable}, and \textit{strongly stable}, as well as developing algorithms to determine whether each type of matching exists in a given instance and find one if it does.

The Stable Roommates problem is a generalization of the stable marriage problem to a non-bipartite setting where people can be matched regardless of gender. Unlike in the original problem, a stable matching in this setting does not always exist. Irving \cite{irving} developed an $O(n^2)$ algorithm to find a stable matching or report that none exists in a given instance, where $n$ is the number of people.

On the other hand, the noncrossing matching problem in a bipartite graph where the points lie on two parallel lines, each containing $n$ points, was encountered in many real-world situations such as in VLSI layout design \cite{kajitami}. In the special case where each point is adjacent to exactly one point on the opposite line, Fredman \cite{fredman} presented an $O(n \log n)$ algorithm to find a maximum size noncrossing matching by computing the length of the longest increasing subsequence (LIS). Widmayer and Wong \cite{widmayer} developed another algorithm that runs in $O(k+(n-k)\log(k+1))$ time, where $k$ is the size of the solution. This algorithm has the same worst-case runtime as Fredman's, but runs faster in most general cases.

In a more general case where each point may be adjacent to more than one point, Malucelli et al. \cite{malucelli} developed an algorithm to find a maximum size noncrossing matching. The algorithm runs in either $O(m \log\log n)$ or $O(m+\min{(nk,m \log k)})$ time depending on implementation, where $m$ is the number of edges and $k$ is the size of the solution. In the case where each edge has a weight, they also showed that the algorithm can be adapted to find a maximum weight noncrossing matching with $O(m \log n)$ runtime.

\subsection{Our Contribution}
In this paper, we constructively prove that a weakly stable noncrossing matching always exists in any instance by developing an $O(n^2)$ algorithm to find one in a given instance.

\section{Preliminaries}
In this setting, we have a set of $n$ men $m_1,...,m_n$ represented by points lying on a vertical line in this order from top to bottom, and a set of $n$ women $w_1,...,w_n$ represented by points lying on another parallel line in this order from top to bottom. Only people of opposite genders can be matched with each other, and each person can be matched with at most one other person. A \textit{matching} is a set of disjoint man-woman pairs.

For a person $a$ and a matching $M$, define $M(a)$ to be the person matched with $a$ (for convenience, let $M(a) = null$ for an unmatched person $a$). For each person $a$, let $P_a$ be the preference list of $a$ containing people of opposite gender to $a$ as his/her acceptable partners in decreasing order of preference. Note that a preference list does not have to contain all $n$ people of opposite gender. Throughout this paper, we assume that the preference lists are \textit{strict} (containing no tie involving two or more people). Also, let $r_a(b)$ be the rank of a person $b$ in $P_a$, with the first entry having rank 1, the second entry having rank 2, and so on (for convenience, let $r_a(null) = \infty$ and treat $null$ as the last entry appended to the end of $P_a$, as being matched is always better than being unmatched). A person $a$ is said to prefer a person $b$ to a person $c$ if $r_a(b) < r_a(c)$.

A pair of edges cross each other if they intersect in the interior of both segments. Formally, an edge $(m_i,w_x)$ crosses an edge $(m_j,w_y)$ if and only if $(i-j)(x-y) < 0$. A matching is called \textit{noncrossing} if it does not contain any pair of crossing edges.

The following are the formal definitions of a \textit{blocking pair} given in the original stable marriage problem, and a \textit{noncrossing blocking pair} introduced here.

\begin{definition}
A \textit{blocking pair} w.r.t. a matching $M$ is a pair $(m,w)$ of a man and a woman that are not matched with each other, but $m$ prefers $w$ to $M(m)$ and $w$ prefers $m$ to $M(w)$.
\end{definition}

\begin{definition}
A \textit{noncrossing blocking pair} w.r.t. a matching $M$ is a blocking pair w.r.t. $M$ that does not cross any edge in $M$. 
\end{definition}

We also introduce two types of stable noncrossing matchings, distinguished as weakly and strongly stable.

\begin{definition}
A matching $M$ is called a \textit{weakly stable noncrossing matching} (WSNM) if $M$ is noncrossing and does not admit any noncrossing blocking pair.  
\end{definition}

\begin{definition}
A matching $M$ is called a \textit{strongly stable noncrossing matching} (SSNM) if $M$ is noncrossing and does not admit any blocking pair.  
\end{definition}

Note that an SSNM is a matching that is both noncrossing and stable, while a WSNM is ``stable" in a weaker sense as it may admit a blocking pair, just not a noncrossing one.

An SSNM may not exist in some instances. For example, in an instance of two men and two women, with $P_{m_1} = (w_2,w_1), P_{m_2} = (w_1,w_2), P_{w_1} = (m_2,m_1)$, and $P_{w_2} = (m_1,m_2)$, the only stable matching is $\{ (m_1,w_2), (m_2,w_1) \}$, and its two edges do cross each other. On the other hand, the above instance has two WSNMs $\{ (m_1,w_2)\}$ and $\{(m_2,w_1) \}$. It also turns out that a WSNM always exists in every instance. Throughout this paper, we focus on the proof of existence of a WSNM by developing an algorithm to find one.

\section{Our Algorithm}
\subsection{Outline}
Without loss of generality, for each man $m$ and each woman $w$, we assume that $w$ is in $m$'s preference list if and only if $m$ is also in $w$'s preference list (otherwise we can simply remove the entries that are not mutual from the lists). Initially, every person is unmatched.

Our algorithm uses proposals from men to women similarly to the Gale– Shapley algorithm in \cite{gale}, but in a more constrained way. With $M$ being the current noncrossing matching, when a woman $w$ receives a proposal from a man $m$, if she prefers her current partner $M(w)$ to $m$, she rejects $m$; if she is currently unmatched or prefers $m$ to $M(w)$, she dumps $M(w)$ and accepts $m$.

Consider a man $m$ and a woman $w$ not matched with each other. An entry $w$ in $P_m$ has the following possible states:
\begin{enumerate}[label*=\arabic*.]
	\item \textbf{\textit{accessible}} (to $m$), if $(m,w)$ does not cross any edge in $M$;
	\begin{enumerate}[label*=\arabic*.]
		\item \textbf{\textbf{\textit{available}}} (to $m$), if $w$ is accessible to $m$, and is currently unmatched or matched with a man she likes less than $m$, i.e. $m$ is going to be accepted if he proposes to her (for convenience, if $w$ is currently matched with $m$, we also call $w$ accessible and available to $m$).
		\item \textbf{\textit{unavailable} }(to $m$), if $w$ is accessible to $m$, but is currently matched with a man she likes more than $m$, i.e. $m$ is going to be rejected if he proposes to her;
	\end{enumerate}
	\item \textbf{\textit{inaccessible}} (to $m$), if $w$ is not accessible to $m$;
\end{enumerate}

For a man $m$, if every entry in $P_m$ before $M(m)$ is either inaccessible or unavailable, then we say that $m$ is \textit{stable}; otherwise (there is at least one available entry before $M(m)$) we say that $m$ is \textit{unstable}.

The main idea of our algorithm is that, at any point, if there is at least one unstable man, we pick the topmost unstable man $m_i$ (the unstable man $m_i$ with least index $i$) and perform the following operations.
\begin{enumerate}
	\item Let $m_i$ \textit{dump} his current partner $M(m_i)$ (if any), i.e. remove $(m_i,M(m_i))$ from $M$, and let him propose to the available woman $w_j$ that he prefers most.
	\item Let $w_j$ \textit{dump} her current partner $M(w_i)$ (if any), i.e. remove $(M(w_j),w_j)$ from $M$, and let her accept $m_i$'s proposal.
	\item Add the new pair $(m_i,w_j)$ to $M$.
\end{enumerate}
We repeatedly perform such operations until every man becomes stable. Note that throughout the algorithm, every proposal will result in acceptance and $M$ will always be noncrossing since men propose only to women available to them.

\subsection{Proof of Correctness}
First, we will show that if our algorithm stops, then the matching $M$ given by the algorithm must be a WSNM.

Assume, for the sake of contradiction, that $M$ admits a noncrossing blocking pair $(m_i,w_j)$. That means $m_i$ prefers $w_j$ to his current partner $M(m_i)$, $w_j$ prefers $m_i$ to her current partner $M(w_j)$, and $(m_i,w_j)$ does not cross an edge in $M$, thus the entry $w_j$ in $P_{m_i}$ is available and is located before $M(m_i)$. However, by the description of our algorithm, the process stops when every man becomes stable, which means there cannot be an available entry before $M(m_i)$ in $P_{m_i}$, a contradiction. Therefore, we can conclude that our algorithm gives a WSNM as a result whenever it stops.

However, it is not trivial that our algorithm will eventually stop. In contrast to the Gale–Shapley algorithm in the original stable marriage problem, in this problem a woman is not guaranteed to get increasingly better partners throughout the process because a man can dump a woman too if he later finds a better available woman previously inaccessible to him (due to having an edge obstructing them). In fact, it is actually the case where the process may not stop if at each step we pick an arbitrary unstable man instead of the topmost one. For example, in an instance of two men and two women with $P_{m_1} = (w_2,w_1), P_{m_2} = (w_1,w_2), P_{w_1} = (m_1,m_2), P_{w_2} = (m_2,m_1)$, the order of picking $m_1, m_2, m_2, m_1, m_1, m_2, m_2, m_1, ...$ results in the process continuing forever, with the matching $M$ looping between $\{(m_1,w_2)\}$, $\{(m_2,w_2)\}$, $\{(m_2,w_1)\}$, and $\{(m_1,w_1)\}$ at each step.

We will prove that our algorithm must eventually stop and evaluate its worst-case runtime after we introduce the explicit implementation of the algorithm in the next subsection.

\subsection{Implementation}
To implement the above algorithm, we have to consider how to efficiently find the topmost unstable man at each step in order to perform the operations on him. Of course, a straightforward way to do this is to update the state of every entry in every man's preference list after each step, but that method will be very inefficient. Instead, we introduce the following scanning method.

Throughout the algorithm, we do not know exactly the set of all unstable men, but we instead keep a set $S$ of men that are ``possibly unstable." Initially, the set $S$ contains all men, i.e. $S = \{m_1,m_2,...,m_n\}$, and at each step we maintain the set $S$ of the form $\{m_i,m_{i+1},...,,m_n\}$ for some $i \in [n]$ (that means $m_1,m_2,...,m_{i-1}$ are guaranteed to be stable at that time). Note that in the actual implementation, we can store only the index of the topmost man in $S$ instead of the whole set. At each step, we scan the topmost man $m_i$ in $S$ and check whether $m_i$ is stable. If $m_i$ is already stable, then we simply skip him by removing $m_i$ from $S$ and moving to scan the next man in $S$. If $m_i$ is unstable, then $m_i$ is indeed the topmost unstable man we want, so we perform the operations on $m_i$. Note that the operations may cause some men to become unstable, so after that we have to add all men that are possibly affected by the operations back to $S$. The details of the scanning and updating processes are as follows.

During the scan of $m_i$, let $m_{prev}$ be the \textit{matched} man closest to $m_i$ that lies above him, and let $w_{first} = M(m_{prev})$ (we let $w_{first} = w_1$ if there is no $m_{prev}$). Also, let $m_{next}$ be the \textit{matched} man closest to $m_i$ that lies below him, and let $w_{last} = M(m_{next})$ (we let $w_{last} = w_n$ if there is no $m_{next}$). Observe that matching $m_i$ with anyone lying above $w_{first}$ will cross the edge $(m_{prev},w_{first})$, and matching $m_i$ with anyone lying below $w_{last}$ will cross the edge $(m_{next},w_{last})$. Therefore, the range of all women accessible to $m_i$ ranges exactly from $w_{first}$ to $w_{last}$, hence the range of all women available to $m_i$ ranges from either $w_{first}$ or $w_{first+1}$ (depending on whether $w_{first}$ prefers $m_i$ to $m_{prev}$) to either $w_{last}$ or $w_{last-1}$ (depending on whether $w_{last}$ prefers $m_i$ to $m_{next}$). See Fig. 1.

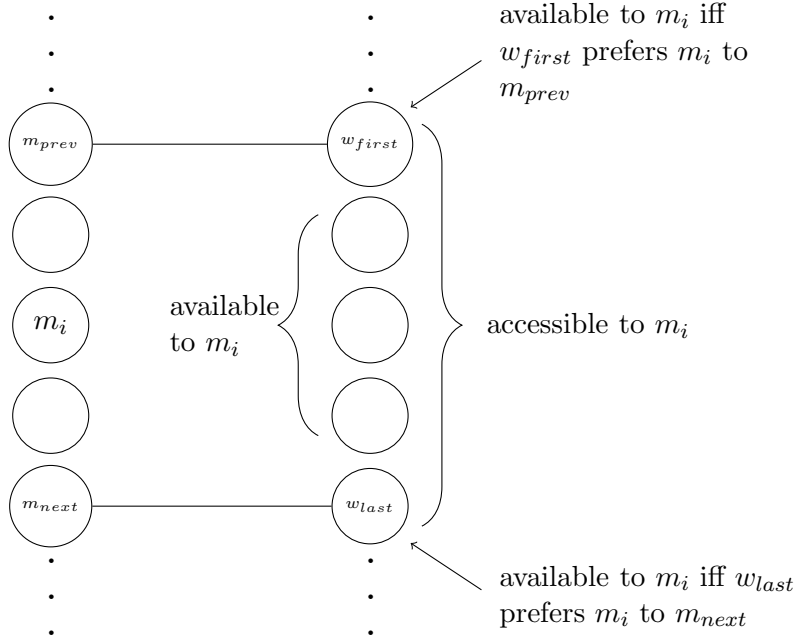
\begin{figure}
\begin{minipage}{0.5\textwidth}
\begin{tikzpicture}[node distance=1.2cm, auto]
		\node [] (a0) {};
    \node [vertex, right of=a0, xshift=2cm] (a1) {\tiny $m_{prev}$};
    \node [vertex, below of=a1] (a2) {};
    \node [vertex, below of=a2] (a3) {$m_i$};
    \node [vertex, below of=a3] (a4) {};
    \node [vertex, below of=a4] (a5) {\tiny $m_{next}$};
    \node [above of=a1, text width=0.1cm] {\textbf{. . .}};
    \node [below of=a5, text width=0.1cm] {\textbf{. . .}};
    \node [vertex, right of=a1, xshift=3cm] (b1) {\tiny $w_{first}$};
    \node [vertex, below of=b1] (b2) {};
    \node [vertex, below of=b2] (b3) {};
    \node [vertex, below of=b3] (b4) {};
    \node [vertex, below of=b4] (b5) {\tiny $w_{last}$};
    \node [above of=b1, text width=0.1cm] {\textbf{. . .}};
    \node [below of=b5, text width=0.1cm] {\textbf{. . .}};
    
    \node [right of=b1, xshift=-0.8cm, yshift=0.4cm] (c1) {};
    \node [right of=c1, yshift=0.8cm] (d1) {};
    \node [right of=d1, text width=4cm, xshift=0.9cm] {available to $m_i$ iff $w_{first}$ prefers $m_i$ to $m_{prev}$};
    \node [right of=b5, xshift=-0.8cm, yshift=-0.4cm] (c5) {};
    \node [right of=c5, yshift=-0.8cm] (d5) {};
    \node [right of=d5, text width=4cm, xshift=0.9cm] {available to $m_i$ iff $w_{last}$ prefers $m_i$ to $m_{next}$};
    \node [left of=b2, xshift=0.8cm, yshift=0.4cm] (c2) {};
    \node [left of=b4, xshift=0.8cm, yshift=-0.4cm] (c4) {};
    
    \draw [decorate,decoration={brace,amplitude=15pt,raise=8pt}] (c1) -- (c5) node[midway,xshift=1cm]{accessible to $m_i$};
    \draw [decorate,decoration={brace,amplitude=15pt,raise=8pt}] (c4) -- (c2) node[midway,xshift=-0.6cm,text width=1.5cm]{available to $m_i$};
    
    \path [line] (a1) -- node{} (b1);
    \path [line] (a5) -- node{} (b5);
    \path [arrow] (d1) -- node{} (c1);
    \path [arrow] (d5) -- node{} (c5);
\end{tikzpicture}
\end{minipage}
\caption{Accessible and available women to $m_i$}
\end{figure}

Then in the available range, $m_i$ selects the woman $w_j$ that he prefers most.

\textbf{Case 1:} $w_j$ does not exist or $m_i$ is currently matched with $w_j$.

That means $m_i$ is currently stable, so we can skip him. We remove $m_i$ from $S$ and proceed to scan $m_{i+1}$ in the next step (called a \textit{downward jump}).

\textbf{Case 2:} $w_j$ exists and $m_i$ is not currently matched with $w_j$

That means $m_i$ is indeed the topmost unstable man we want, so we perform the operations on him by letting $m_i$ propose to $w_j$ and dump his current partner (if any).

\textbf{Case 2.1:} $m_{prev}$ exists and $w_j = w_{first}$.

That means $w_{first}$ dumps $m_{prev}$ to get matched with $m_i$, which leaves $m_{prev}$ unmatched and he may possibly become unstable. Furthermore, $m_{prev+1},$ $m_{prev+2}, ..., m_{i-1}$ as well as $m_i$ himself may also possibly become unstable since they now gain access to women lying above $w_{first}$ previously inaccessible to them (if $w_{first} \neq w_1$). On the other hand, $m_1, m_2, ..., m_{prev-1}$ clearly remain stable, hence we add $m_{prev}, m_{prev+1}, ..., m_{i-1}$ to $S$ and proceed to scan $m_{prev}$ in the next step (called an \textit{upward jump}).

\textbf{Case 2.2:} $m_{prev}$ does not exist or $w_j \neq w_{first}$.

\textbf{Case 2.2.1:} $m_i$ is currently matched and $w_j$ lies geometrically below $M(m_i)$.

Then, $m_{prev}, m_{prev+1}, ..., m_{i-1}$ (or $m_1,m_2,...m_{i-1}$ if $m_{prev}$ does not exist) may possibly becomes unstable since they now gain access to women between $M(m_i)$ and $w_j$ previously inaccessible to them. Therefore, we perform the upward jump to $m_{prev}$ (or to $m_1$ if $m_{prev}$ does not exist), adding $m_{prev}, m_{prev+1}, ...,$ $m_{i-1}$ (or $m_1,m_2,...m_{i-1}$) to $S$ and proceed to scan $m_{prev}$ (or $m_1$) in the next step, except when $m_i = m_1$ that we perform the downward jump to $m_2$.

It turns out that this case is impossible, which we will prove in the next subsection.

\textbf{Case 2.2.2:} $m_i$ is currently unmatched or $w_j$ lies geometrically above $M(m_i)$.

Then all men lying above $m_i$ clearly remain stable (because the sets of available women to $m_1,...,m_{i-1}$ either remain the same or become smaller). Also, $m_i$ now becomes stable as well (because $m_i$ selects a woman he prefers most in the available range), except in the case where $w_j = w_{last}$ (because the edge $(m_{next},w_{last})$ is removed and $m_i$ now has access to women lying below $w_{last}$ previously inaccessible to him). Therefore, we perform the downward jump, removing $m_i$ from $S$ and moving to scan $m_{i+1}$ in the next step, except when $w_j = w_{last}$ that we have to scan $m_i$ again in the next step (this exception, however, turns out to be impossible, which we will prove in the next subsection).

We scan the men in this way until $S$ becomes empty (see Example 1). By the way we add all men that may possibly become unstable after each step back to $S$, at any step $S$ is guaranteed to contain the topmost unstable man.

\begin{example}
Consider an instance of three men and three women with the following preference lists.

\begin{figure}[H]
	\centering
	\begin{minipage}{0.3\textwidth}
		$\boldsymbol{m_1:} \hspace{0.1cm} w_3, w_1, w_2$ \\
		$\boldsymbol{m_2:} \hspace{0.1cm} w_2, w_3, w_1$ \\
		$\boldsymbol{m_3:} \hspace{0.1cm} w_2, w_1, w_3$ \\
	\end{minipage}
	\begin{minipage}{0.3\textwidth}
		$\boldsymbol{w_1:} \hspace{0.1cm} m_3, m_2, m_1$ \\
		$\boldsymbol{w_2:} \hspace{0.1cm} m_3, m_2, m_1$ \\
		$\boldsymbol{w_3:} \hspace{0.1cm} m_3, m_2, m_1$ \\
	\end{minipage}
\end{figure}
		
Our algorithm will scan the men in the following order and output a matching $M = \{(m_2,w_1),(m_3,w_2)\}$, which is a WSNM.

\begin{center}
\begin{table}[H]
	\begin{tabular}{|c|l|l|l|}
		\hline
		\textbf{\thead{Step}} & \textbf{\thead{Process}} & \textbf{\thead{$M$ at the end\\ of step}} & \textbf{\thead{$S$ at the end\\ of step}} \\ \hline
		0 & & $\O$ & $\{m_1,m_2,m_3\}$ \\ \hline
		1 & scan $m_1$, add $(m_1,w_3)$ & $\{(m_1,w_3)\}$ & $\{m_2,m_3\}$ \\ \hline
		2 & scan $m_2$, add $(m_2,w_3)$, remove $(m_1,w_3)$ & $\{(m_2,w_3)\}$ & $\{m_1,m_2,m_3\}$ \\ \hline
		3 & scan $m_1$, add $(m_1,w_1)$ & $\{(m_1,w_1),(m_2,w_3)\}$ & $\{m_2,m_3\}$ \\ \hline
		4 & scan $m_2$, add $(m_2,w_2)$, remove $(m_2,w_3)$ & $\{(m_1,w_1),(m_2,w_2)\}$ & $\{m_3\}$ \\ \hline
		5 & scan $m_3$, add $(m_3,w_2)$, remove $(m_2,w_2)$ & $\{(m_1,w_1),(m_3,w_2)\}$ & $\{m_2,m_3\}$ \\ \hline
		6 & scan $m_2$, add $(m_2,w_1)$, remove $(m_1,w_1)$ & $\{(m_2,w_1),(m_3,w_2)\}$ & $\{m_1,m_2,m_3\}$ \\ \hline
		7 & scan $m_1$ & $\{(m_2,w_1),(m_3,w_2)\}$ & $\{m_2,m_3\}$ \\ \hline
		8 & scan $m_2$ & $\{(m_2,w_1),(m_3,w_2)\}$ & $\{m_3\}$ \\ \hline
		9 & scan $m_3$ & $\{(m_2,w_1),(m_3,w_2)\}$ & $\O$ \\ \hline
	\end{tabular}
\end{table}
\end{center}
\end{example}

\subsection{Observations}
First, we will prove the following lemma about the algorithm described in the previous subsection.

\begin{lemma} \label{lem00}
During the scan of a man $m_i$, if $m_i$ is currently matched, then $m_i$ does not propose to any woman lying geometrically below $M(m_i)$.
\end{lemma}

\begin{proof}
We call a situation when a man $m_i$ proposes to a woman lying geometrically below $M(m_i)$ a \textit{downward switch}. Assume, for the sake of contradiction, that a downward switch occurs at least once during the whole algorithm. Suppose that the first downward switch occurs at step $s$, when a man $m_i$ is matched to $w_k = M(m_i)$ and proposes to $w_j$ with $j > k$. We have $m_i$ prefers $w_j$ to $w_k$.

Consider the step $t < s$ when $m_i$ proposed to $w_k$ (if $m_i$ proposed to $w_k$ multiple times, consider the most recent one). At step $t$, $w_j$ must be inaccessible or unavailable to $m_i$ (otherwise he would choose $w_j$ instead of $w_k$), meaning that there must be an edge $(m_p,w_q)$ with $p > i$ and $k < q < j$ obstructing them in the inaccessible case, or an edge $(m_p,w_q)$ with $p > i$, $q = j$, and $w_j$ preferring $m_p$ to $m_i$ in the unavailable case.

We define a \textit{dynamic edge} $e$ as follows. First, at step $t$ we set $e = (m_p,w_q)$. Then, throughout the process we update $e$ by the following method: whenever the endpoints of $e$ cease to be partners of each other, we update $e$ to be the edge joining the endpoint that dumps his/her partner with his/her new partner. Formally, suppose that $e$ is currently $(m_x,w_y)$. If $m_x$ dumps $w_y$ to get matched with $w_{y'}$, we update $e$ to be $(m_x,w_{y'})$; if $w_y$ dumps $m_x$ to get matched with $m_{x'}$, we update $e$ to be $(m_{x'},w_y)$.

By this updating method, the edge $e$ will always exist after step $t$, but may change over time. Observe that from step $t$ to step $s$, we always have $x > i$ because of the existence of $(m_i,w_k)$. Moreover, before step $s$, if $m_x$ dumps $w_y$ to get matched with $w_{y'}$, by the assumption that a downward switch did not occur before step $s$, we have $y' < y$, which means the index of the women's side of $e$'s endpoints never increases. Consider the edge $e=(m_x,w_y)$ at step $s$, we must have $x > i$ and $y \leq q \leq j$. If $y < j$, then the edge $e$ obstructs $m_i$ and $w_j$, making $w_j$ inaccessible to $m_i$. If $y = j$, that means $w_j$ never got dumped since step $t$, so she got only increasingly better partners, thus $w_j$ prefers $m_x$ to $m_i$, making $w_j$ unavailable to $m_i$. Therefore, in both cases $m_i$ could not propose to $w_j$, a contradiction. Hence, a downward switch cannot occur in our algorithm.
\end{proof}

Lemma \ref{lem00} shows that a woman cannot get her partner stolen by any woman that lies below her, which is equivalent to the following corollary.

\begin{corollary} \label{cor00}
If a man $m_i$ dumps a woman $w_j$ to propose to a woman $w_k$, then $k < j$.
\end{corollary}

It also implies that Case 2.2.1 in the previous subsection never occurs. Therefore, the only case where an upward jump occurs is Case 2.1 ($m_{prev}$ exists and $m_i$ proposes to $w_{first}$). We will now prove the following lemma.

\begin{lemma} \label{lem0}
During the scan of a man $m_i$, if $m_{next}$ exists, then $m_i$ does not propose to $w_{last}$.
\end{lemma}

\begin{proof}
Assume, for the sake of contradiction, that $m_i$ proposes to $w_{last}$. Since $m_{next}$ exists, this proposal obviously cannot occur in the very first step of the algorithm. Consider a man $m_k$ we scanned in the previous step right before scanning $m_i$.

\textbf{Case 1:} $m_k$ lies below $m_i$, i.e. $k>i$.

In order for the upward jump from $m_k$ to $m_i$ to occur, $m_i$ must have been matched with a woman but got her stolen by $m_k$ in the previous step. However, $m_{i+1}, m_{i+2}, ..., m_{next-1}$ are all currently unmatched (by the definition of $m_{next}$), so the only possibility is that $m_k = m_{next}$, and thus his partner that got stolen was $w_{last}$. Therefore, we can conclude that $w_{last}$ prefers $m_{next}$ to $m_i$, which means $w_{last}$ is currently unavailable to $m_i$, a contradiction.

\textbf{Case 2:} $m_k$ lies above $m_i$, i.e. $k<i$.

The jump before the current step was a downward jump, but since $m_{next}$ has been scanned before, an upward jump over $m_i$ must have occurred at some point before the current step. Consider the most recent upward jump over $m_i$ before the current step. Suppose than it occurred at the end of step $t$ and was a jump from $m_{k'}$ to $m_j$, with $k'>i$ and $j<i$. In order for this jump to occur, $m_j$ must have been matched with a woman but got her stolen by $m_{k'}$ at step $t$. However, $m_{i+1}, m_{i+2}, ..., m_{next-1}$ are all currently unmatched (by the definition of $m_{next}$), so the only possibility is that $m_{k'} = m_{next}$, and thus $m_j$'s partner that got stolen was $w_{last}$. We also have $m_{j+1}, m_{j+2}, ..., m_{next-1}$ were all unmatched during step $t$ (otherwise $w_{last}$ would be inaccessible to $m_{k'}$ ), and $w_{last}$ prefers $m_{next}$ to $m_j$.

Now, consider the most recent step before step $t$ in which we scanned $m_i$. Suppose it occurred at step $s$. During step $s$, $m_j$ was matched with $w_{last}$ and $w_{last}$ was accessible to $m_i$. However, $m_i$ was still left unmatched after step $s$ (otherwise an upward jump over $m_i$ at step $t$ could not occur), meaning that $w_{last}$ must be unavailable to him back then due to $w_{last}$ preferring $m_j$ to $m_i$. Therefore, we can conclude that $w_{last}$ prefers $m_{next}$ to $m_i$, thus $w_{last}$ is currently unavailable to $m_i$, a contradiction.
\end{proof}

Lemma \ref{lem0} shows that a man cannot get his partner stolen by any man lying above him, or equivalent to the following corollary.

\begin{corollary} \label{cor0}
If a woman $w_j$ dumps a man $m_i$ to accept a man $m_k$, then $k > i$.
\end{corollary}

\subsection{Proof of Finiteness}

Now, we will show that the position of each woman's partner can only move downward throughout the process, which guarantees the finiteness of the number of steps in the entire process.

\begin{lemma} \label{lem2}
After a woman $w_j$ ceases to be a partner of a man $m_i$, she cannot be matched with any man $m_{i'}$ with $i' \leq i$ afterwards.
\end{lemma}

\begin{proof}
Suppose that $w_j$'s next partner (if any) is $m_a$. It is sufficient to prove that $a > i$. First, consider the situation when $m_i$ and $w_j$ cease to be partners.

\textbf{Case 1:} $w_j$ dumps $m_i$.

This means $w_j$ dumps $m_i$ to get matched with $m_a$ right away. By Corollary \ref{cor0}, we have $a > i$ as required.

\textbf{Case 2:} $m_i$ dumps $w_j$.

Suppose that $m_i$ dumps $w_j$ to get matched with $w_k$. By Corollary \ref{cor00}, we have $k < j$.

\textbf{Case 2.1:} $m_i$ never gets dumped afterwards.

That means $m_i$ will only get increasingly better partner, and the position of his partner can only move upwards (by Corollary \ref{cor00}), which means $w_j$ cannot be matched with $m_i$ again, or any man lying above $m_i$ afterwards due to having an edge $(m_i,M(m_i))$ obstructing. Therefore, $m_a$ must lie below $m_i$, i.e. $a > i$.

\textbf{Case 2.2:} $m_i$ gets dumped afterwards.

Suppose that $m_i$ first gets dumped by $w_y$ at step $s$. By Corollary \ref{cor00}, we have $y \leq k < j$ (because $m_i$ only gets increasingly better partners before getting dumped). Also suppose that $w_y$ dumps $m_i$ in order to get matched with $m_x$. By Corollary \ref{cor0}, we have $x > i$. Similarly to the proof of Lemma \ref{cor00}, consider a dynamic $e$ first set to be $(m_x,w_y)$ at step $s$. We have the index of the men's side of $e$'s endpoints never decreases, and that of the women's side never increases. Therefore, since step $s$, there always exists an edge $(m_x,w_y)$ with $x>i$ and $y<j$, obstructing $w_j$'s access to $m_i$ and all men lying above him. Therefore, $m_a$ must lie below $m_i$, i.e. $a > i$.
\end{proof}

\subsection{Runtime Analysis}
Consider any upward jump from $m_i$ to $m_k$ with $i > k$ that occurs right after $m_i$ stole $w_j$ from $m_k$. We call such a jump \textit{associated to} $w_j$, and it has \textit{size} $i-k$.

For any woman $w_j$, let $U_j$ be the sum of the sizes of all upward jumps associated to $w_j$. From Lemma \ref{lem2}, we know that the position of $w_j$'s partner can only move upward throughout the process, so we have $U_j \leq n-1$. Therefore, the sum of the sizes of all upwards jumps is $\sum_{j=1}^n U_j \leq n(n-1) = O(n^2)$. Since the scan starts at $m_1$ and ends at $m_n$, the total number of downward jumps equals to the sum of the sizes of all upward jumps plus $n-1$, hence the total number of steps in the whole algorithm is $O(n^2)$.

For each $m_i$, we keep an array of size $n$, with the $j$th entry storing the rank of $w_j$ in $P_{m_i}$. Each time we scan $m_i$, we query the minimum rank of available women, which is a consecutive range in the array. Using an appropriate range minimum query (RMQ) data structure such as the one introduced by Fischer \cite{fischer}, we can perform the scan with $O(n)$ preprocessing time per array and $O(1)$ query time. Therefore, the total runtime of our algorithm is $O(n^2)$.

In conclusion, we proved that our developed algorithm is correct and can be implemented in $O(n^2)$ time, which also implicitly proves the existence of a WSNM in any instance.

\begin{theorem} \label{thm1}
A weakly stable noncrossing matching exists in any instance with $n$ men and $n$ women with strict preference lists.
\end{theorem}

\begin{theorem} \label{thm2}
There is an $O(n^2)$ algorithm to find a weakly stable noncrossing matching in an instance with $n$ men and $n$ women with strict preference lists.
\end{theorem}

\begin{remark}
Our algorithm does not require the numbers of men and women to be equal. In the case that there are $n_1$ men and $n_2$ women, the algorithm works similarly with $O(n_1n_2)$ runtime. Also, in the case that people's preference lists are not strict, we can modify the instance by breaking ties in an arbitrary way. Clearly, a WSNM in the modified instance will also be a WSNM in the original one (because every noncrossing blocking pair in the original instance will also be a noncrossing blocking pair in the modified instance).
\end{remark}

\section{Discussion}
In this paper, we constructively prove that a WSNM always exists in any instance by developing an $O(n^2)$ algorithm to find one. Note that the definition of a WSNM allows multiple answers with different sizes for an instance. For example, in an instance of three men and three women, with $P_{m_1} = (w_3,w_1,w_2), P_{m_2} = (w_1,w_2,w_3), P_{m_3} = (w_2,w_3,w_1), P_{w_1} = (m_2,m_3,m_1), P_{w_2} = (m_3,m_1,m_2)$, and $P_{w_3} = (m_1,m_2,m_3)$, both $\{ (m_1,w_3)\}$ and $\{ (m_2,w_1),(m_3,w_2)\}$ are WSNMs, but our algorithm only outputs the first one with smaller size. A possible future work is to develop an algorithm to find a WSNM with maximum size in a given instance, which seems to be a naturally better answer then other WSNMs as it satisfies more people. Another possible future work is to develop an algorithm to determine whether an SSNM exists in a given instance, and to find one if it does.

Other interesting problems include investigate the noncrossing matching in the geometric version of the stable roommates problem where people can be matched regardless of gender. The most basic and natural setting of this problem is where people are represented by points arranged on a circle. A possible future work is to develop an algorithm to determine whether a WSNM or an SSNM exists in a given instance, and to find one if it does.

\end{document}